\documentclass[11pt]{article}  

\usepackage{amsthm,amssymb,amsmath,graphicx,xcolor,times}

\usepackage[margin=2.5cm]{geometry}

\usepackage{natbib}
\usepackage[pdftex,pagebackref,colorlinks,linkcolor=black,filecolor = black, citecolor = black, urlcolor  = black]{hyperref}

\usepackage{prettyref}
\newcommand{\savehyperref}[2]{\texorpdfstring{\hyperref[#1]{#2}}{#2}}

\newrefformat{eq}{\savehyperref{#1}{\textup{(\ref*{#1})}}}
\newrefformat{lem}{\savehyperref{#1}{Lemma~\ref*{#1}}}
\newrefformat{def}{\savehyperref{#1}{Definition~\ref*{#1}}}
\newrefformat{thm}{\savehyperref{#1}{Theorem~\ref*{#1}}}
\newrefformat{cor}{\savehyperref{#1}{Corollary~\ref*{#1}}}
\newrefformat{cha}{\savehyperref{#1}{Chapter~\ref*{#1}}}
\newrefformat{sec}{\savehyperref{#1}{Section~\ref*{#1}}}
\newrefformat{app}{\savehyperref{#1}{Appendix~\ref*{#1}}}
\newrefformat{tab}{\savehyperref{#1}{Table~\ref*{#1}}}
\newrefformat{fig}{\savehyperref{#1}{Figure~\ref*{#1}}}
\newrefformat{hyp}{\savehyperref{#1}{Hypothesis~\ref*{#1}}}
\newrefformat{alg}{\savehyperref{#1}{Algorithm~\ref*{#1}}}
\newrefformat{sdp}{\savehyperref{#1}{SDP~\ref*{#1}}}
\newrefformat{rem}{\savehyperref{#1}{Remark~\ref*{#1}}}
\newrefformat{item}{\savehyperref{#1}{Item~\ref*{#1}}}
\newrefformat{step}{\savehyperref{#1}{step~\ref*{#1}}}
\newrefformat{conj}{\savehyperref{#1}{Conjecture~\ref*{#1}}}
\newrefformat{fact}{\savehyperref{#1}{Fact~\ref*{#1}}}
\newrefformat{prop}{\savehyperref{#1}{Proposition~\ref*{#1}}}
\newrefformat{claim}{\savehyperref{#1}{Claim~\ref*{#1}}}
\newrefformat{relax}{\savehyperref{#1}{Relaxation~\ref*{#1}}}
\newrefformat{red}{\savehyperref{#1}{Reduction~\ref*{#1}}}
\newrefformat{part}{\savehyperref{#1}{Part~\ref*{#1}}}
\newrefformat{prob}{\savehyperref{#1}{Problem~\ref*{#1}}}

\DeclareMathOperator {\Var}  {Var}

\newcommand {\set}   [1] {\left\{ #1 \right\}}
\newcommand {\brc}   [1] {\left(#1\right)}

\newcommand {\Exp}       {\mathbb{E}}
\newcommand {\Prob}  [1] {\Pr \brc{#1 }}

\newcommand {\E}     [1] {\Exp\left[#1\right]}

\newcommand {\given}{\mid}

\DeclareMathOperator*{\poly}{poly}

\newcommand{\iprod}[2]{\langle #1, #2\rangle}
\newcommand{\Nin}{N^{\sf in}}
\newcommand{\Nout}{N^{\sf out}}
\newcommand{\Abs}[1]{\left|#1\right|}

\newcommand{\agdeg}{\eta^H_{\text{max}}}

\DeclareMathOperator {\sdpcost}{sdp-cost}

\DeclareMathOperator {\diam}{diam}

\newtheorem{theorem}{Theorem}[section]
\newtheorem{lemma}[theorem]{Lemma}

\newtheorem{definition}[theorem]{Definition}

\newtheorem{problem}[theorem]{Problem}

\newtheorem*{question*}{Question}
\newtheorem*{definition*}{Definition}

\newtheorem{claim}{Claim}[section]
\newtheorem{remark}{Remark}[section]

\newcommand{\Ecut}[1]{E_{cut}(#1)}

\title{Approximation Algorithms for Hypergraph Small Set Expansion \\ and Small Set Vertex Expansion}
\author{Anand Louis\thanks{Supported by Santosh Vempala's NSF award CCF-1217793.}
\\Georgia Tech \\ anandl@gatech.edu \and Yury Makarychev\thanks{Supported by NSF CAREER award CCF-1150062 and NSF award IIS-1302662.}\\TTIC\\ yury@ttic.edu}
\date{}
\begin{document}
\maketitle
\begin{abstract}
The expansion of a hypergraph, a natural extension of the notion of expansion in graphs, is defined
as the minimum over all cuts in the hypergraph of the ratio of the number of the hyperedges cut to the 
size of the smaller side of the cut.
We study the Hypergraph Small Set Expansion problem, which, for a parameter $\delta \in (0,1/2]$, asks to
compute the cut having the least expansion while having at most $\delta$ fraction of the vertices 
on the smaller side of the cut. We present two algorithms.
Our first algorithm gives an $\tilde O(\delta^{-1} \sqrt{\log n})$ approximation. The second algorithm finds a set with expansion
$\tilde O(\delta^{-1}(\sqrt{d_{\text{max}}r^{-1}\log r\, \phi^*} + \phi^*))$ in
a $r$--uniform hypergraph with maximum degree $d_{\text{max}}$ (where
$\phi^*$ is the expansion of the optimal solution).
Using these results, we also obtain algorithms for the Small Set Vertex Expansion problem:
we get an
$\tilde O(\delta^{-1} \sqrt{\log n})$ approximation algorithm and an algorithm that finds a set with vertex expansion
$O\left(\delta^{-1}\sqrt{\phi^V \log d_{\text{max}} } + \delta^{-1} \phi^V\right)$
(where $\phi^V$ is the vertex expansion of the optimal solution).
 
For $\delta=1/2$, Hypergraph Small Set Expansion is equivalent to the hypergraph expansion problem. In this case, our approximation factor of $O(\sqrt{\log n})$ for expansion in hypergraphs matches the corresponding approximation
factor for expansion in graphs due to \cite{ARV}.
\end{abstract}

\section{Introduction} 
\label{sec:intro}
The expansion of a hypergraph, a natural extension of the notion of expansion in graphs, is defined
as follows.
\begin{definition}[Hypergraph Expansion]
Given  a hypergraph $H = (V, E)$ on $n$ vertices (each edge $e\in E$ of $H$ is a subset of vertices),
we say that an edge $e \in E$ is cut by a set $S$ if $e\cap S\neq \varnothing$ and $e\cap \bar{S} \neq \varnothing$
(i.e. some vertices in $e$ lie in $S$ and some vertices lie outside of $S$). We denote the set of edges cut by $S$ by 
$\Ecut{S}$.
%\end{definition}
%\begin{definition}
The expansion $\phi (S)$ of a set $S\subset V$ ($S\neq \varnothing$, $S\neq V$) in a hypergraph $H = (V, E)$ is
defined as
\[ \phi(S) = \frac{|\Ecut{S}|}{\min (|S|,|\bar S|)}. \]
\end{definition}

Hypergraph expansion  and related hypergraph partitioning problems are of 
immense practical importance, having applications in parallel and distributed computing (\citet{ca99}), VLSI circuit 
design and computer architecture (\citet{kaks99,gglp00}), scientific computing (\cite{dbh06}) and other areas.
Inspite of this, there has no't been much theoretical work on them.
In this paper, we study a generalization of the Hypergraph Expansion prolbem, namely the Hypergraph Small Set Expansion problem.
\begin{problem}[Hypergraph Small Set Expansion Problem]
Given a hypergraph $H = (V,E)$ and a parameter $\delta \in (0,1/2]$, 
the Hypergraph Small Set Expansion problem (H-SSE) is to find a set $S \subset V$ of size at most $\delta n$ that minimizes $\phi(S)$.
The value of the optimal solution to H-SSE is called the small set expansion of $H$.
That is, for $\delta \in (0,1/2]$, the small set expansion $\phi^*_{H,\delta}$  of a hypergraph $H = (V,E)$ is defined as
\[ \phi^*_{H,\delta} = \min_{\substack{S\subset V\\ 0 < |S| \leq \delta n}} \phi(S). \]
\end{problem}
\noindent Note that for $\delta = 1/2$, the Hypergraph Small Set Expansion Problem is the 
Hypergraph Expansion Problem.

%\smallskip

% We are given a hypergraph $H = (V,E)$ on $n$ vertices and a parameter $\delta \in (0,1/2)$. Our goal is to find a set of size
% at most $\delta n$ that has minimal expansion. The problem generalizes the Small Set Expansion problem in graphs. 
Small Set Expansion in graphs  has attracted a lot of attention recently. 
The problem was introduced by \citet{RS}, who showed that it is closely related to the
Unique Games problem. 
\cite{RST} designed an algorithm for SSE that finds a set of size $O(\delta n)$ with expansion  $O(\sqrt{\phi^* d \log (1/\delta)})$ in $d$ regular graphs
(where $\phi^*$ is the expansion of the optimal solution).
Later Bansal, Feige, Krauthgamer, Makarychev, Nagarajan, Naor, and Schwartz (2011) 
gave a $O(\sqrt{\log n \log (1/\delta)})$ approximation algorithm for the problem.

We present analogs of the results of~\cite{BFK}
and \cite{RST} for hypergraphs.
Our first result is an $\tilde O(\delta^{-1}\sqrt{\log n})$ 
approximation algorithm\footnote{The $\tilde O$--notation hides a $\log \delta^{-1}\log\log \delta^{-1}$ term.}
for H-SSE (see  \prettyref{thm:hsse}).
Our second result is an algorithm that finds a set with expansion
at most 
$\tilde{O}\left(\delta^{-1} \left(\sqrt{ 
d_{\text{max}} \frac{\log r}{r}
\phi^*_{H,\delta}} + \phi^*_{H,\delta} \right)\right)
$ if $H$ is an $r$--uniform hypergraph with maximum degree
$d_{\text{max}}$ (see \prettyref{thm:hsse-2}; the result also
applies to non-uniform hypergraphs, see \prettyref{thm:hsse-3}).

We note that H-SSE can be reduced to SSE (small set expansion in graphs) if all hyperedges have bounded size.
Let $r$ be the size of the largest hyperedge in $H$. Construct an auxiliary graph $F$ on $V$ as follows: 
pick a vertex in each hyperedge $e$ and connect it in $F$ to all other vertices of $e$ (i.e. replace $e$ with a star).
Then solve SSE in the graph $F$. It is easy to see that if we solve SSE using an $\alpha$ approximation algorithm,
then we get $(r-1) \alpha$ approximation for H-SSE. This approach gives $O(\sqrt{\log n \log (1/\delta)})$ approximation
if $r$ is bounded. However, if $H$ is an arbitrary hypergraph, we only get an  
$O(n\sqrt{\log n \log (1/\delta)})$ approximation. The goal of this paper is to give an approximation guarantee valid
for hypergraphs with hyperedges of arbitrary size.
We now formally state our main results.
\begin{theorem}
\label{thm:hsse}
There is a randomized polynomial-time approximation algorithm for the Hypergraph Small Set Expansion problem that
given a hypergraph $H = (V,E)$, and parameters $\varepsilon \in (0,1)$ and $\delta \in (0,1/2)$,
finds a set $S \subset V$ of size at most $(1+\varepsilon)\delta n$ such that 
\[
\phi(S) \leq O_{\varepsilon}\left(\delta^{-1}\log \delta^{-1}\log\log \delta^{-1}\cdot \sqrt{\log n}  \cdot  \phi^*_{H,\delta}\right)
=\tilde O_{\varepsilon}\left(\delta^{-1} \sqrt{\log n}  \, \phi^*_{H,\delta}\right),
\]
(where the constant in the $O$ notation depends polynomially on $1/\varepsilon$).
That is, the algorithm gives $O(\sqrt{\log n})$ approximation when $\delta$ and $\varepsilon$ are fixed.
\end{theorem}
\noindent We state our second result, \prettyref{thm:hsse-2}, 
for $r$-uniform hypergraphs. We present and prove a more general
\prettyref{thm:hsse-3} that applies to any hypergraphs in \prettyref{sec:algoHSSE-2}. 
\begin{theorem}\label{thm:hsse-2}
There is a randomized polynomial-time algorithm for the Hypergraph Small Set Expansion problem that
given an $r$--uniform hypergraph $H = (V,E)$ with maximum degree $d_{\text{max}}$, and parameters $\varepsilon \in (0,1)$ and $\delta \in (0,1/2)$
finds a set $S \subset V$ of size at most $(1+\varepsilon)\delta n$ such that 
\[
\phi(S) \leq \tilde O_{\varepsilon} \left(\delta^{-1} \left(\sqrt{ 
d_{\text{max}}\frac{\log r}{r}
\phi^*_{H,\delta}} + \phi^*_{H,\delta} \right)\right).
\]
\end{theorem}

%We present an  algorithm for H-SSE with approximation factor
%$O(\delta^{-c}\sqrt{\log n})$, where $c$ is an absolute constant.

Our algorithms for H-SSE are bi-criteria approximation algorithms in that they output a set $S$ of
size at most $(1+\varepsilon)\delta n$. 
We note that this is similar to the algorithm of~\cite{BFK} for SSE, which also finds a set of size at most 
$(1+\varepsilon) \delta n$ rather than a set of size at most $\delta n$. The algorithm of~\cite{RST} 
finds a set of size $O(\delta n)$.
The approximation factor of our first algorithm does not depend on the size of hyperedges in the input hypergraph. 
It has the same dependence on $n$ as the algorithm of~\cite{BFK} for SSE. However, the dependence on $1/\delta$ is quasi-linear; whereas it is logarithmic in the algorithm of~\cite{BFK}. In fact, we show that the integrality gap of 
the standard SDP relaxation for H-SSE is at least linear in $1/\delta$ (\prettyref{thm:gap}).
The approximation guarantee of our second algorithm is analogous to that
of the algorithm of~\cite{RST}.

\paragraph{Small Set Vertex Expansion.}
Our techniques can also be used to obtain an approximation algorithm for Small Set Vertex Expansion (SSVE)
in graphs.
\begin{problem}[Small Set Vertex Expansion Problem]
Given graph $G = (V,E)$, the vertex expansion of a set $S \subset V$ is defined as
\[ \phi^V(S) = \frac{ |\{  u \in \bar{S} : \exists\, v \in S \textrm{ such that } \{u,v \} \in E   \}  |  }{|S|}  \]
Given a parameter $\delta \in (0,1/2]$, 
the Small Set Vertex Expansion  problem (SSVE) is to find a set $S \subset V$ of size at most $\delta n$ that minimizes $\phi^V(S)$.
The value of the optimal solution to SSVE is called the small set vertex expansion of $G$.
That is, for $\delta \in (0,1/2]$, the small set expansion $\phi^V_{G,\delta}$  of a graph $G = (V,E)$ is defined as
\[ \phi^V_{G,\delta} = \min_{\substack{S\subset V\\ 0 < |S| \leq \delta n}} \phi^V(S). \]
\end{problem}
Small Set Vertex Expansion recently gained interest due to its connection to obtaining 
sub\-exponen\-tial-time, constant factor approximation algorithms for many combinatorial problems like Sparsest Cut and Graph Coloring (\cite{ag11,lrv12}). Using  a reduction from vertex expansion in graphs to hypergraph expansion, we can get
an approximation algorithm for SSVE having the same approximation guarantee as that for H-SSE.

\begin{theorem}
\label{thm:hyper-vert-exp}
There exist absolute constants $c_1, c_2 \in \mathbb{R}^+$ such that for 
 every  graph $G=(V,E)$, %of maximum degree $d$, 
there exists a polynomial time computable hypergraph 
$H = (V',E')$ %having the hyperedges of cardinality at most $d+1$ 
such that
$$ c_1 \phi^*_{H,\delta} \leq \phi^V_{G,\delta} \leq c_2 \phi^*_{H,\delta}. $$
Also, $\agdeg \leq \log_2 (d_{\text{max}} +1)$, where $d_{\text{max}}$ is the maximum degree of $G$ (where 
$\agdeg$ is defined in \prettyref{def:agdeg}).
\end{theorem}
%For the sake of completeness, we give a proof of this theorem 
%due to \cite{Louis} in Appendix \prettyref{sec:vertex-to-hyper}. 
\noindent From this theorem, \prettyref{thm:hsse} and \prettyref{thm:hsse-3} we immediately
get algorithms for SSVE.
\begin{theorem}[Corollary to \prettyref{thm:hsse} and \prettyref{thm:hyper-vert-exp}]
There is a randomized polynomial-time approximation algorithm for the Small Set Vertex Expansion problem that
given a graph $G = (V,E)$, and parameters $\varepsilon \in (0,1)$ and $\delta \in (0,1/2))$
finds a set $S \subset V$ of size at most $(1+\varepsilon)\delta n$ such that 
\[
\phi^V(S) \leq O_{\varepsilon}\left(\sqrt{\log n}\,\delta^{-1}\log \delta^{-1}
\log\log \delta^{-1}  \cdot  \phi^V_{G,\delta}\right)
,
\]
That is, the algorithm gives $O(\sqrt{\log n})$ approximation when $\delta$ and $\varepsilon$ are fixed.
\end{theorem}
\begin{theorem}[Corollary to \prettyref{thm:hsse-3} and \prettyref{thm:hyper-vert-exp}]\label{thm:v-expansion-2}
There is a randomized polynomial-time algorithm for the Small Set Vertex Expansion problem that
given a graph $G = (V,E)$ of maximum degree $d_{max}$,  parameters $\varepsilon \in (0,1)$ and $\delta \in (0,1/2)$
finds a set $S \subset V$ of size at most $(1+\varepsilon)\delta n$ such that 
\begin{align*}
\phi^V(S) &\leq O_{\varepsilon}\left(\sqrt{\phi^V_{G,\delta} \log d_{\text{max}}}\cdot\delta^{-1}\log \delta^{-1}
\log\log \delta^{-1}  + \delta^{-1} \phi^V_{G,\delta}\right) \\
&=
\tilde O_{\varepsilon}\left(\delta^{-1}\sqrt{\phi^V_{G,\delta} \log d_{\text{max}} } + \delta^{-1} \phi^V_{G,\delta}\right)
.
\end{align*}
\end{theorem}
We note that the Small Set Vertex Expansion problem for $\delta = 1/2$ is just the Vertex Expansion problem. In that case,
\prettyref{thm:v-expansion-2} gives the same approximation guarantee as the algorithm of~\cite{lrv13}.

\smallskip

\noindent\textbf{Techniques.} 
Our general approach to solving H-SSE is similar to the approach of~\citet{BFK}. We recall how the algorithm of \citet{BFK}
for (graph) SSE works. The algorithm solves a semidefinite programming relaxation for SSE and gets an SDP solution. The SDP solution
assigns a vector $\bar u$ to each vertex $u$. Then the algorithm generates \textit{an orthogonal separator}. Informally, an orthogonal separator $S$  with distortion $D$ is a random subset of vertices such that 
\begin{itemize}
\item[(a)]
If $\bar u$ and $\bar v$ are close to each other then the probability that $u$ and $v$ are separated by $S$ is small; namely, it is at most $\alpha D \|\bar u - \bar v\|^2$, where $\alpha$ is a normalization factor such that $\Prob{u \in S} = \alpha \|\bar u\|^2$. 
\item[(b)] If the angle between $\bar u$ and $\bar v$
is larger than a certain threshold, then the probability that both $u$ and $v$ are in $S$ is much smaller than the probability that one of them is in $S$. 
\end{itemize}
\citet{BFK} showed
that condition (b) together with SDP constraints implies that $S$ is of size at most $(1+\varepsilon)\delta n$
with sufficiently high probability. Then condition (a) implies that the expected number of cut edges is at most $D$ 
times the SDP value. That means that $S$ is a $D$--approximate solution to SSE.

 If we run this algorithm on an instance of H-SSE, we will still find a set of size at most $(1+\varepsilon)\delta n$, but the cost of the solution might be very high.
Indeed, consider a hyperedge $e$. Even though every two vertices $u$ and $v$ in $e$ are \textit{unlikely} to be separated by $S$, at least one pair out of $\binom{|e|}{2}$ pairs of vertices is quite likely to be separated by $S$;
hence, $e$ is quite likely to be cut by $S$. To deal with this problem, we develop \textit{hypergraph orthogonal separators}.
In the definition of a hypergraph orthogonal separator, we strengthen condition (a) by requiring that
a hyperedge $e$ is cut by $S$ with small probability if all vertices in $e$ are close to each other. Specifically,
%we introduce two types of separators, hypergraph orthogonal separators and $\ell_2$--$\ell_2^2$ hypergraph orthogonal separators. 
we require that 
\begin{equation}
\Prob{e \text{ is cut by } S} \leq \alpha D \max_{u,v\in e}\|\bar u-\bar v\|^2. \label{eq:h-orth-sep1}
\end{equation}
%This condition is analogous to condition (a). 
We show that there is a hypergraph orthogonal separator with distortion
proportional to $\sqrt{\log n}$ (the distortion also depends on parameters of the orthogonal separator).
Plugging this hypergraph orthogonal separator in the algortihm of~\citet{BFK}, we get~\prettyref{thm:hsse}.
We also develop another variant of hypergraph orthogonal separators, $\ell_2$--$\ell_2^2$ orthogonal separators.
An $\ell_2$--$\ell_2^2$ orthogonal separator with $\ell_2$--distortion $ D_{\ell_2}(r)$ and $\ell_2^2$--distortion
$D_{\ell_2^2}$ satisfies the following condition\footnote{It may look strange that we have two terms in the bound.
One may expect that we can either have only term $D_{\ell_2^2}\max_{u,v\in e}\|\bar u-\bar v\|^2$ (as in the previous definition)
or only term $D_{\ell_2}(|e|) \cdot \min_{w\in E} \|\bar w\| \cdot \max_{u,v\in e}\|\bar u-\bar v\|$. However, the latter is not possible --- there is no
$\ell_2$--$\ell_2^2$ separator with $D_{\ell_2^2} = 0$.}
\begin{equation}
\Prob{e \text{ is cut by } S} \leq \alpha D_{\ell_2}(|e|) \cdot \min_{w\in E} \|\bar w\| \cdot \max_{u,v\in e}\|\bar u-\bar v\| +  \alpha D_{\ell_2^2}\cdot\max_{u,v\in e}\|\bar u-\bar v\|^2.
\label{eq:h-orth-sep2}
\end{equation}
We show that there is an $\ell_2$-$\ell_2^2$ hypergraph orthogonal separator whose $\ell_2$ and $\ell_2^2$ distortions do not depend on $n$
(in contrast, there is no hypergraph orthogonal separator whose distortion does not depend on $n$). This result yields \prettyref{thm:hsse-2}.

We now give a brief conceptual overview of our construction of hypergraph orthogonal separators. We use the framework developed in~\citep*[Section 4.3]{CMM} for (graph) orthogonal separators. 
For simplicity, we ignore vector normalization steps in this overview; we do not explain how we take into account vector lengths. Note, however, that these normalization steps are crucial.
We first design a procedure that partitions
the hypergraph into two pieces (the procedure labels every vertex with either $0$ or $1$). In a sense, each set $S$ in the partition 
is a ``very weak'' hypergraph orthogonal  separator. It satisfies property \prettyref{eq:h-orth-sep1} with $D_0 \sim
\sqrt{\log n}\log\log (1/\delta)$ and $\alpha_0 = 1/2$
and a weak variant of property (b): if the angle between vectors $\bar u$ and $\bar v$
is larger than the threshold then events $u\in S$ and $v\in S$ are ``almost'' independent.
We repeat the procedure $l = \log_2 (1/\delta) +O(1)$ times and obtain a partition of graph into $2^l=O(1/\delta)$ pieces.
Then we randomly choose one set $S$ among them; this set $S$ is our hypergraph orthogonal separator. Note that
that by running the procedure many times we decrease exponentially in $l$ the probability that two vertices, as in condition (b), belong to $S$. So condition (b) holds for $S$.
Also, we affect the distortion in~\prettyref{eq:h-orth-sep1} in two ways.
First, the probability that the edge is cut increases by a factor of $l$. That is, we get
$\Prob{e \text{ is cut by } S} \leq l \times \alpha_0 D_0 \max_{u,v\in e}\|\bar u-\bar v\|^2$.
Second, the probability that we choose a vertex $u$ goes down from $\|\bar u\|^2/2$ to $\Omega(\delta) \|\bar u\|^2$ since roughly speaking we choose one set $S$ among $O(1/\delta)$ possible sets. That is, the parameter $\alpha$ of $S$ is $\Omega(\delta)$.
Therefore, $\Prob{e \text{ is cut by } S} \leq \alpha (\alpha_0 l D_0/\alpha) \max_{u,v\in e}\|\bar u-\bar v\|^2$.
That is, we get a hypergraph orthogonal separator with distortion $(\alpha_0 l D_0/\alpha) \sim \tilde O(\delta^{-1}\sqrt{\log n})$.
The construction of $\ell_2^2$ orthogonal separators is similar but a bit more technical.

\medskip

\noindent\textbf{Organization.} 
We present our SDP relaxation and introduce our main technique, hypergraph orthogonal separators, in \prettyref{sec:prelim}. 
We describe our first algorithm for H-SSE in \prettyref{sec:algoHSSE}, and then 
%The algorithm uses hypergraph orthogonal separators.
describe an algorithm that generates hypergraph orthogonal separators in \prettyref{sec:orthogonal}. 
%We give the algorithm for Small Set Vertex Expansion in \prettyref{sec:vertex-to-hyper}.
We define $\ell_2$--$\ell_2^2$ hypergraph orthogonal separators,
%(another variant of hypergraph orthogonal separators), 
give an algorithm that generates them, and then present our second algorithm for H-SSE in \prettyref{sec:orthogonal-l2-l22} and \prettyref{sec:algoHSSE-2}.
Finally, we show a simple SDP integrality
gap for H-SSE in \prettyref{sec:SDPgap}. This integrality gap also gives a lower bound on the quality of $m$-orthogonal separators.
We give a proof of \prettyref{thm:hyper-vert-exp} in \prettyref{sec:vertex-to-hyper}.

\section{Preliminaries}\label{sec:prelim}
\subsection{SDP Relaxation for Hypergraph Small Set Expansion}
We use the SDP relaxation for H-SSE shown in \prettyref{fig:SDP}. There is an SDP variable $\bar u$ for every vertex $u\in V$.
\begin{figure}\label{fig:SDP}
%\OpenFrame
\begin{align}
\text{minimize }& \sum_{e\in E} \max_{u,v\in e} \|\bar u - \bar v\|^2\\
\text{subject to:} \notag\\
&\sum_{v\in V} \langle \bar u, \bar v\rangle \leq \delta n \cdot \|\bar u\|^2 && \text{for every } u\in V\\
&\sum_{u\in V} \|\bar u\|^2 = 1\\
&\|\bar u-\bar v\|^2 + \|\bar v-\bar w\|^2 \geq \|\bar u-\bar w\|^2 && \text{for every } u,v,w\in V \label{eq:triangle-ineq1}\\
& 0 \leq \langle \bar u, \bar v \rangle \leq \| u \|^2 && \text{for every } u,v\in V. \label{eq:triangle-ineq2}
\end{align}
%\CloseFrame
\caption{SDP relaxation for H-SSE}
\end{figure}
Every combinatorial solution $S$ (with $|S| \leq \delta n$) defines the corresponding (intended) SDP solution: 
$\bar u = \frac{e}{\sqrt{|S|}}$, if  $u\in S$; $\bar u = 0$, otherwise,
\iffalse
\[
\bar u = 
\begin{cases}
\frac{e}{\sqrt{|S|}} , \text{ if } u\in S,\\
0, \text{ otherwise,} 
\end{cases}
\]
\fi
where $e$ is a fixed unit vector. It is easy to see that this solution satisfies all SDP constraints.
Note that $\max_{u,v\in e} \|\bar u - \bar v\|^2$ is equal to $1/|S|$, if $e$ is cut, and to 0, otherwise.
Therefore, the objective function equals 
\[
\sum_{e\in E} \max_{u,v\in e} \|\bar u - \bar v\|^2 = \sum_{e \in \Ecut{S}} \frac{1}{|S|} = \frac{\Ecut{S}}{S}=\phi(S).
\]
Thus our SDP for H-SSE is indeed a relaxation.

\subsection{Hypergraph Orthogonal Separators}
The main technical tool for proving \prettyref{thm:hsse} is  \textit{hypergraph orthogonal separators}. 
\textit{Orthogonal separators} were introduced by~\citet*{CMM} (see also 
%Bansal, Feige, Krauthgamer, Makarychev, Nagarajan, Naor, and Schwartz (2011), 
\citet{BFK}, \citet*{LM}, and \citet*{MM}) and were previously used for solving Unique Games and various graph partitioning problems.
In this paper, we extend the technique of orthogonal separators to hypergraphs and introduce hypergraph orthogonal separators.
We then use hypergraph orthogonal separators to solve H-SSE.
In \prettyref{sec:orthogonal-l2-l22}, we introduce another
version of hypergraph orthogonal separators, $\ell_2$--$\ell_2^2$
hypergraph orthogonal separators, and then use them to prove
\prettyref{thm:hsse-2} and \prettyref{thm:hsse-3}. 

\begin{definition}[Hypergraph Orthogonal Separators]\label{def:h-othogonal-sep}
Let $\set{\bar u:u\in V}$ be a set of vectors in the unit ball that satisfy $\ell_2^2$--triangle 
inequalities~\prettyref{eq:triangle-ineq1} and \prettyref{eq:triangle-ineq2}. We say that a random set $S\subset V$ is  \textit{a hypergraph $m$-orthogonal separator}
with distortion $D \geq 1$, probability scale $\alpha > 0$, and separation threshold $\beta \in (0,1)$ if it satisfies the following properties.
\begin{enumerate}
\item For every $u\in V$, 
\[\Pr(u\in S) = \alpha \|\bar u\|^2.\]
\item For every $u$ and $v$ such that $\| \bar u - \bar v \|^2 \geq \beta \min(\|\bar u\|^2, \|\bar v\|^2)$
% (i.e., the angle between $\bar u$ and $\bar v$ is at least ...).
\[ 
\Prob{u\in S \text{ and } v\in S} \leq \alpha \frac{\min (\|\bar u\|^2, \|\bar v\|^2)}{m}.
\]
\item For every $e\subset V$,
\[
\Prob{e \text{ is cut by }{S}} \leq \alpha D \max_{u,v\in e} \| \bar u - \bar v\|^2.
\]
\end{enumerate}
\end{definition}
\noindent 
The definition of a hypergraph $m$-orthogonal separator is similar to that of a (graph) 
\mbox{$m$-orthogonal} separator:  a random set $S$ is an {$m$-orthogonal} separator if it satisfies properties~1, 2, and property~$3'$,
which is property 3 restricted to edges $e$ of size 2.
\begin{enumerate}
\item[$3'$.] For every $(u,v)$,
$
\Prob{e \text{ is cut by }{S}} \leq \alpha D \| \bar u - \bar v\|^2
$.
\end{enumerate}
In this paper, we design an algorithm that generates a hypergraph $m$-orthogonal separator 
with distortion $O_{\beta}(\sqrt{\log n}\cdot m\log m\log\log m)$. 
We note that the distortion of \textit{any} hypergraph orthogonal separator
must depend on $m$ at least linearly (see \prettyref{sec:SDPgap}).
We remark that there are two constructions of (graph) orthogonal separators, ``orthogonal separators via $\ell_1$''
and ``orthogonal separators via $\ell_2$'', with distortions, 
$O_{\beta}(\sqrt{\log n} \log m)$ and $O_{\beta}(\sqrt{\log n \log m})$, respectively (presented in
\citep*{CMM}). 
Our construction of hypergraph orthogonal separators uses the framework of orthogonal separators via $\ell_1$.
We prove the following theorem in \prettyref{sec:orthogonal}. 

\begin{theorem}
\label{thm:hyper-orth-sep}
There is a polynomial-time randomized algorithm that given a set of vertices $V$, a set of vectors $\{\bar u\}$ satisfying $\ell_2^2$--triangle inequalities 
\prettyref{eq:triangle-ineq1} and \prettyref{eq:triangle-ineq2}, parameters
$m \geq 2$ and $\beta\in(0,1)$, generates a hypergraph $m$-orthogonal separator with probability scale $\alpha  \geq 1 / n$ 
and distortion $D = O(\beta^{-1} m\log m\log\log m\times \sqrt{\log n})$. 
\end{theorem}

\section{Algorithm for Hypergraph Small Set Expansion}\label{sec:algoHSSE}
In this section, we present our algorithm for Hypergraph Small Set Expansion. Our algorithm
uses hypergraph orthogonal spearators that we describe in \prettyref{sec:orthogonal}.
We use the approach of~\citet{BFK}. Suppose that we are given a polynomial-time algorithm that generates
hypergraph $m$-orthogonal separators with distortion $D(m, \beta)$ (with probability scale $\alpha > 1/\poly(n)$). We show how to 
get a  $D^* = 4 D(4/(\varepsilon \delta), \varepsilon/4)$ approximation for H-SSE.

\begin{theorem}\label{thm:HSSE-proof}
There is a randomized polynomial-time approximation algorithm for the Hypergraph Small Set Expansion problem that
given a hypergraph $H = (V,E)$, and parameters $\varepsilon \in (0,1)$ and $\delta \in (0,1/2)$
finds a set $S \subset V$ of size at most $(1+\varepsilon)\delta n$ such that 
\[\phi(S) \leq 4 D(4/(\varepsilon \delta), \varepsilon/4)  \cdot  \phi^*_{H,\delta} .\]
\end{theorem}
\begin{proof}
We solve the SDP relaxation for H-SSE and obtain an SDP solution $\set{\bar u}$. Denote the SDP value by $\sdpcost$.
Consider a hypergraph orthogonal separator $S$
with $m = 4/(\varepsilon \delta)$ and $\beta = \varepsilon/4$.
Define a set $S'$: 
\[
S' = 
\begin{cases}
S,\ \text{if } |S| \leq (1+\varepsilon) \delta n,\\
\varnothing,\ \text{otherwise.}
\end{cases}
\]
Clearly, $|S'| \leq (1+\varepsilon) \delta n$.
\citet{BFK} showed that
$\Pr(u \in S') \in \bigl[\frac{\alpha}{2} \, \|\bar u\|^2,\alpha 
\|\bar u\|^2\bigr] $
for every $u\in V$ (see also Theorem A.1 in \citep{MM}).
Note that 
\[\Prob{S' \text{ cuts edge } e} \leq \Prob{S \text{ cuts edge } e}
\leq \alpha D^* \max_{u,v\in e} \|\bar u - \bar v\|^2.\]
where $D^*$ denotes $D(4/(\varepsilon \delta), \varepsilon/4)$ for the sake of brevity.
Let $Z = |S'| - \frac{|\Ecut{S'}|}{4D^* \cdot \sdpcost}$.
We have,
\begin{align*}\E{Z} &= \E{|S'|} - \frac{\E{ |\Ecut{S'}|} }{4D^* \cdot \sdpcost} \geq \sum_{u \in V} \frac{\alpha}{2} \cdot \|\bar u\|^2 -  
\frac{\sum_{e\in E} \alpha D^* \max_{u,v\in e}\|\bar u - \bar v\|^2}{4D^* \cdot \sdpcost}  \\
&= \frac{\alpha}{2} - \frac{1}{4D^* \cdot \sdpcost} \times \alpha D^* \sdpcost= \frac{\alpha}{4}.
\end{align*}
Since $Z \leq |S'| \leq (1+\varepsilon)\delta n < n$ (always), by Markov's inequality, we have $\Prob{Z > 0} \geq \alpha/(4n)$ and
hence
\[ \Prob{|\Ecut{S'}|/|S'| < 4D^* \cdot \sdpcost} \geq \alpha /(4n).  \]

We sample $S$ independently $4n/\alpha$ times and return the first set $S'$ such that 
$\frac{|\Ecut{S'}|}{|S'|} < 4D^* \cdot \sdpcost$.
This gives a set $S'$ such that $|S'| \leq (1+\varepsilon)\delta n$, and 
$\phi(S') \leq 4D^*  \phi^*_{H,\delta}$. The algorithm succeeds (finds such a set $S'$) with a constant probability. 
By repeating the algorithm $n$ times,
we can make the success probability exponentially close to $1$.
\end{proof}

In \prettyref{sec:orthogonal}, we describe how to generate an $m$-hypergraph orthogonal separator with distortion $D = O\bigl(\sqrt{\log n} \times \beta^{-1} m\log m\log\log m \bigr)$. That gives us an algorithm for H-SSE with approximation factor $O_{\varepsilon}\left(\delta^{-1} \log \delta^{-1} \log\log \delta^{-1}\times \sqrt{\log n}\right)$.

\section{Generating Hypergraph Orthogonal Separators}\label{sec:orthogonal}
%We prove that orthogonal separators presented in~\citet*[Section 4.3]{CMM}
%are \textit{hypergraph} orthogonal separators.
%Since orthogonal 
%To verify that an $m$-orthogonal separator is a hypergraph
%$m$-orthogonal separators, it suffices to check that it satisfies property $3$ since
%it $m$-orthogonal separators satisfy properties 1 and 2 (as well as property $3'$ for all pairs of vertices $(u,v)$).
%We therefore verify know that constructions presented in ~\citet*{CMM} satisfy property 3.

\iffalse
\begin{remark}
Observe that every $m$-orthogonal separator with distortion $D$ is a hypergraph orthogonal separator with distortion
at most $(r-1) D$ (where $r$ is the size of the largest hyperedge). Indeed, consider an arbitrary hyperedge $e$ of size at most $r$.
Choose $u\in e$. We have
\[
\Prob{e \in \Ecut{S}} = \Prob{\text{for some }v\in e: (u,v) \text{ is cut by } S} \leq \sum_{v\in e, v\neq u}
\Prob{(u,v) \text{ is cut by } S}\leq (r-1) \times \alpha  D  \| \bar u - \bar v\|^2.
\]
In this section, we get hypergraph orthogonal separators whose distortion does not depend on $r$.
\end{remark}
\fi

In this section, we present an algorithm that generates a hypergraph $m$-orthogonal separator. At the high level, the algorithm is similar to the algorithm
for generating orthogonal separators from Section~4.3 in~\citep{CMM}.
We use a different procedure for generating words $W(u)$
(see below) and set parameters differently; also the analysis of our algorithm
is different.

In our algorithm, we use a ``normalization'' map $\varphi$ from~\citep{CMM}.
Map $\varphi$ maps a set $\set{\bar u}$ of vectors satisfying $\ell_2^2$--triangle inequalities 
\prettyref{eq:triangle-ineq1} and \prettyref{eq:triangle-ineq2} to $\mathbb{R}^n$. It has the following properties.
\begin{enumerate}
\item  
For all vertices $u$, $v$, $w$,
%If vectors $\bar u,\bar v, \bar w$ satisfy $\ell_2^2$--triangle inequality $\|\bar u - \bar v\|_2^2 + \|\bar v - \bar w\|_2^2 \geq \|\bar u - \bar w\|_2^2$ then
$
\|\varphi(\bar u) - \varphi(\bar v)\|_2^2 + \|\varphi(\bar v)
- \varphi(\bar w)\|_2^2 \geq \|\varphi(\bar u) - \varphi(\bar w)\|_2^2.
$
\item For all nonzero
vertices $u$ and $v$,
$\iprod{\varphi(\bar u)}{\varphi(\bar v)} = \frac{\iprod{\bar u}{\bar v}}{\max(\|\bar u\|^2,\|\bar v\|^2)}.$
%(If $\bar u =\bar v  = 0$ then $\varphi (\bar u) = \varphi (\bar u) =0$ and $\iprod{\varphi(u)}{\varphi(v)} = 0$.)
\item In particular, for every $\bar u \neq 0$, $\|\varphi(\bar u)\|_2^2 = \iprod{\varphi(\bar u)}{\varphi(\bar u)}=  1$. Also,
$\varphi(0) = 0$.
\item For all non-zero vectors $\bar u$ and $\bar v$,
$
\|\varphi(\bar u) - \varphi(\bar v)\|_2^2 \leq \frac{2\,\|\bar u - \bar v \|^2}{\max(\|\bar u\|^2, \|\bar v\|^2)}.
$
\end{enumerate}
We also use the following theorem of~\citet*{ALN}  \citep*[see also][]{ARV}.
\begin{theorem}[Arora, Lee, and Naor (2005), Theorem 3.1] \label{thm:aln} 
There exist constants $C \geq 1$ and $p \in (0,    1/4)$
such that for every $n$ unit vectors $x_u$ ($u\in V$), satisfying 
$\ell_2^2$--triangle inequalities \prettyref{eq:triangle-ineq1}, and every $\Delta > 0$, the following holds. 
There exists a random subset $U $ of $V$ such that for every
$u, v \in V$ with $\|x_u - x_v \|^2 \geq \Delta$,
$
\Prob{u \in U \text{ and } d(v, U) \geq  \frac{\Delta}{C \sqrt{\log n}}} \geq p
$,
where $d(v,U) = \min_{u\in U} \|x_u - x_v\|^2$.
\end{theorem}

First we describe an algorithm that randomly assigns each vertex $u$ a symbol,
either 0 or 1. Then we use this algorithm to generate an orthogonal separator.
\begin{lemma}\label{lem:omega}
There is a randomized polynomial-time algorithm that given a finite set $V$, unit vectors $ \varphi(\bar u)$ for $u\in V$ 
satisfying $\ell_2^2$-triangle inequalities
and a parameter $\beta\in(0,1)$, returns a random assignment 
$\omega:V \to \set{0,1}$  that satisfies the following properties.
\begin{itemize}
\item  For every $u$ and $v$ such that $\| \varphi(\bar u) - \varphi(\bar v) \|^2 \geq \beta$,
\[
\Prob{\omega(u)\neq \omega(v)} \geq 2p,
\]
where $p > 0$ is the constant from \prettyref{thm:aln}.
\item For every set $e\subset V$ of size at least 2, 
\[
\Prob{\omega(u)\neq \omega(v) \text{ for some } u,v\in e} \leq O(\beta^{-1}  \sqrt{\log n}\max_{u,v\in e} \|\varphi(\bar u) - \varphi(\bar v)\|^2).
\]
\end{itemize}
\end{lemma}
\begin{proof}
Let $U$ be the random set from \prettyref{thm:aln} for vectors $x_u = \varphi(\bar u)$ and $\Delta = \beta$. 
Choose $t\in (0,1/(C\sqrt{\log n}))$ uniformly at random. 
Let 
\[
\omega(u) = \begin{cases}
0, \text{ if } d(U_i, u) \leq t, \\
1, \text{ otherwise}.
\end{cases}
\]

Consider first vertices $u$ and $v$ such that $\| \varphi(\bar u) - \varphi(\bar v) \|^2 \geq \beta$.
By \prettyref{thm:aln},
\begin{equation*}
\Prob{u\in U \text{ and } d(v,U) \geq \frac{\Delta}{C\sqrt{\log n}}} \geq p \quad \text{and}\quad 
\Prob{v\in U \text{ and } d(u,U) \geq \frac{\Delta}{C\sqrt{\log n}}} \geq p.
\end{equation*}
Note that in the former case, when $u\in U$ and $d(v,U) \geq \frac{\Delta}{C\sqrt{\log n}}$, we have $\omega(u) = 0$ and $\omega(v) = 1$; 
in the latter case, when $v\in U$ and $d(u,U) \geq \frac{\Delta}{C\sqrt{\log n}}$, we have $\omega(v) = 0$ and $\omega(u) = 1$.
Therefore, the probability that $\omega(u) \neq \omega(v)$ is at least $2p$.

Now consider a set $e\subset V$ of size at least 2.
Let 
\[
\tau_m = \min_{w\in e} d(U,\varphi(\bar w)) \quad \text{and}\quad 
\tau_M = \max_{w\in e} d(U,\varphi(\bar w)).\]
 We have,
$
\tau_M -\tau_m \leq \max_{u, v \in e} \|\varphi(\bar u)  - \varphi(\bar v)\|^2 .
$
Note that if $t < \tau_m$ then $\omega(u) = 1$ for all $u\in e$;
if  $t \geq \tau_M$ then $\omega(u) = 0$ for all $u\in e$. Thus $\omega(u) \neq \omega(v)$ for some $u,v\in e$ only if $t \in [\tau_m,\tau_M)$.
Since the probability density of the random variable $t$ is at most $C\sqrt{\log n}$, we get, 
\begin{align*}
\Prob{\exists\, u,v\in e: \omega(u) \neq \omega(v)} &\leq \Prob{t \in [\tau_m,\tau_M)} \leq \frac{C\sqrt{\log n}}{\Delta} \,(\tau_M-\tau_m) %\\&
\leq 
\frac{C\sqrt{\log n}}{\beta}  \max_{u, v \in e} \|\bar u - \bar v \|^2.
\end{align*}
\end{proof}
We now amplify the result of \prettyref{lem:omega}.
\begin{lemma}\label{lem:omega-amplified}
There is a randomized polynomial time algorithm that given $V$, vectors $\varphi(\bar u)$ and $\beta\in(0,1)$ as in \prettyref{lem:omega},
and a parameter $m \geq 2$,
returns a random assignment 
$\omega:V \to \set{0,1}$  such that: %that satisfies the following properties.
\begin{itemize}
\item  For every $u$ and $v$ such that $\| \varphi(\bar u) - \varphi(\bar v) \|^2 \geq \beta$,
\[
\Prob{\tilde\omega(u)\neq \tilde\omega(v)} \geq \frac12 - \frac{1}{\log_2 m}.
\]
\item For every set $e\subset V$ of size at least 2, 
\[
\Prob{\tilde\omega(u)\neq \tilde\omega(v) \text{ for some } u,v\in e} \leq O(\beta^{-1}  \sqrt{\log n}\cdot \log \log m\cdot \max_{u,v\in e} \|\varphi(\bar u) - \varphi(\bar v)\|^2).
\]
\end{itemize}
\end{lemma}
\vspace{-0.4cm}
We independently sample $K = \max\left(\left\lceil \frac{\log_2 \log_2 m}{-\log_2(1-4p)}\right\rceil,1\right)$ assignments $\omega_1,\dots,
\omega_K$, and let $\tilde\omega(u) = \omega_1(u) \oplus \dots \oplus \omega_K(u)$ (where $\oplus$ denotes addition modulo 2). It is easy to see that
the assignment $\tilde \omega$ satisfies the required properties. We give the proof in \prettyref{sec:proof-omega-amplified}.

We are now ready to present our algorithm.

%\OpenFrame
\begin{enumerate}
\item Set $l = \lceil \log_2 m / (1 - \log_2(1+2/\log_2 m)) \rceil = \log_2 m + O(1)$.
\item Sample $l$ independent assignments 
$\tilde\omega_1, \dots, \tilde\omega_l$ using \prettyref{lem:omega-amplified}. 
\item For every vertex $u$, define word $W(u) = \tilde\omega_1(u) \dots \tilde\omega_l(u)\in \set{0,1}^l$.
\item If $n \geq 2^l$, pick a word $W \in \set{0,1}^l$ uniformly at random. If $n < 2^l$, pick a random word $W \in \set{0,1}^l$ 
so that $\Pr_W(W = W(u)) = 1/n$ for every  $u\in V$. This is possible since the number of distinct words constructed in step 3 is at
most $n$ (we may pick a word $W$ not equal to any $W(u)$).

\item Pick $r\in (0,1)$ uniformly at random.

\item Let $S = \set{u \in V: \|\bar u\|^2 \geq r \text{ and } W(u) = W}$.
\end{enumerate}
%\CloseFrame

\begin{theorem}\label{thm:h-ort-sep-distortion-bound}
Random set $S$ is a hypergraph $m$-orthogonal separator with distortion
\[
D = O\bigl(\sqrt{\log n} \times \frac{m\log m\log\log m}{\beta} \bigr),
\]
probability scale $\alpha \geq 1/n$ and separation threshold $\beta$.
\end{theorem}
\begin{proof}
We verify that $S$ satisfies properties 1--3 in the definition of a hypergraph $m$-orthogonal separator with $\alpha = \max (1/2^l, 1/n)$.

\smallskip
\noindent\textbf{Property 1.}
We compute the probability that $u \in S$. Observe that $u\in S$ if and only if $W(u) = W$ and $r\leq \|\bar u\|^2$ (these two events are independent).
If $n \geq 2^l$, the probability that $W = W(u)$ is $1/2^l$
since we choose $W$ uniformly at random from $\set{0,1}^l$;
if $n < 2^l$ the probability is $1/n$. That is,
$\Prob{W = W(u)} = \max(1/2^l, 1/n) = \alpha$.
The probability that $r\leq \|\bar u\|^2$ is  $\|\bar u\|^2$. We conclude that property 1 holds. 

\medskip
\noindent\textbf{Property 2.} Consider two vertices $u$ and $v$ such that
$\| \bar u - \bar v \|^2 \geq \beta \min (\|\bar u\|^2, \|\bar v\|^2)$.
Assume without loss of generality that $\|\bar u\|^2 \leq \|\bar v\|^2$.
Note that $u,v\in S$ if and only if $r \leq \|\bar u\|^2$ and
$W = W(u) = W(v)$.
We first upper bound the probability that $W(u) = W(v)$. We have,
%\begin{equation*}
$
2\iprod{\bar u}{\bar v} = \|\bar u\|^2 + \|\bar v\|^2  - \|\bar u - \bar v\|^2
\leq  (1-\beta) \|\bar u\|^2 + \|\bar v\|^2 \leq  
(2-\beta) \|\bar v\|^2.
$ %\end{equation*}
Therefore,
$2\iprod{\bar u}{\bar v}/\|\bar v\|^2 \leq 2-\beta$.
Hence,
$
\|\varphi(\bar u) - \varphi(\bar v)\|^2 = 2 - 2\iprod{\varphi(\bar u)}{\varphi(\bar v)}=
2 - \frac{2\iprod{\bar u}{\bar v}}{\max(\|\bar u\|^2,\|\bar v\|^2)}
\geq\beta =\Delta.
$
From \prettyref{lem:omega-amplified} we get that 
$\Prob{\tilde\omega_i(u)\neq \tilde\omega_i(v)} \geq \frac12 - \frac{1}{\log_2 m}$ for every $i$.
The probability that $W(u) = W(v)$ is at most
$(\frac12 + \frac{1}{\log_2 m})^l\leq 1/m$. 
We have,
\begin{align*}
\Prob{u\in S, v\in S} &= \Prob{r \leq \min (\|\bar u\|^2, \|\bar v\|^2)}
\times \Prob{W= W(u) = W(v) \given W(u) = W(v)} \times \\
& \phantom{{}={}\times{}}\Prob{W(u) = W(v)}
\leq \min (\|\bar u\|^2, \|\bar v\|^2) \times \alpha \times (1/m),
\end{align*}
as required.

\smallskip
\noindent\textbf{Property 3.}
Let $e$ be an arbitrary subset of $V$, $|e| \geq 2$. Let $\rho_m = \min_{w\in e} \|\bar w\|^2$ and $\rho_M = \max_{w\in e} \|\bar w\|^2$.
Note that
 \[
    \rho_M-\rho_m = \|\bar w_1\|^2  - \|\bar w_2\|^2 \leq \|\bar w_1  -\bar w_2\|^2 \leq \max_{u, v \in e} \|\bar u  -\bar v\|^2,
 \]
for some $w_1,w_2\in e$. Here we used that  SDP constraint~\prettyref{eq:triangle-ineq2} implies that $ \|\bar w_1\|^2  - \|\bar w_2\|^2 \leq \|\bar w_1  -\bar w_2\|^2$.

Let $A = \set{u \in e: \|\bar u \|^2 \geq r}$. Note that $S\cap e = \set{u \in A: W(u) = W}$. Therefore, if $e$ is cut by $S$ then one of the following
events happens.
\begin{itemize}
\item Event ${\cal E}_1$: $A \neq e$ and $S\cap e \neq \varnothing$.
\item Event ${\cal E}_2$: $A = e$ and $A \cap S \neq \varnothing$,   $A \cap S \neq A$.
 \end{itemize}
If ${\cal E}_1$ happens then $r \in [\rho_m, \rho_M]$ since $A \neq e$ and $A\neq \varnothing$. We have, 
\[
\Prob{{\cal E}_1} \leq \Prob{r\in (\rho_m,\rho_M]} \leq |\rho_M - \rho_m|
\leq \max_{u, v \in e} \|\bar u  -\bar v\|^2.
\]
If ${\cal E}_2$ happens then (1) $r \leq \rho_m $ (since $A = e$) and (2) $W(u) \neq W(v)$ for some $u,v\in e$.
The probability that $r \leq \rho_m$ is $\rho_m$. We now upper bound the probability that  $W(u) \neq W(v)$ for some  $u,v\in e$.
For each $i\in\set{1,\dots,l}$,
\begin{align*}
\Prob{\tilde\omega_i(u)\neq \tilde\omega_i(v) \text{ for some } u,v\in e} &\leq
O(\beta^{-1}  \sqrt{\log n}\cdot \log\log m) \max_{u, v \in e}  \|\varphi(\bar u) - \varphi(\bar v)\|^2 \\ 
&\leq  
O(\beta^{-1}  \sqrt{\log n}\cdot \log\log m) \max_{u, v \in e} \frac{2\|\bar u - \bar v \|^2}{\min(\|\bar u\|^2, \|\bar v\|^2)}\\
&\leq  O(\beta^{-1}  \sqrt{\log n}\cdot \log\log m)
\times \rho_m^{-1} \times  \max_{u, v \in e} \|\bar u - \bar v \|^2.
\end{align*}

By the union bound over $i\in\set{1,\dots,l}$, the probability that $W(u) \neq W(v)$ for some $u,v\in e$ is at most 
$O(l\times \beta^{-1}  \sqrt{\log n}\cdot \log\log m)
\times \rho_m^{-1} \times  \max_{u, v \in e} \|\bar u - \bar v \|^2$.
Therefore,
\begin{align*}
\Prob{{\cal E}_2} &\leq \rho_m \times O(l\times \beta^{-1}  \sqrt{\log n}\, \log\log m)
\times \rho_m^{-1} \times  \max_{u, v \in e} \|\bar u - \bar v \|^2\\
&\leq 
O(\beta^{-1}  \sqrt{\log n}\, \log m \log\log m)
\times  \max_{u, v \in e} \|\bar u - \bar v \|^2.
\end{align*}
We get that the probability that $e$ is cut by $S$ is at most
\[
\Prob{{\cal E}_1} + \Prob{{\cal E}_2} \leq O(\beta^{-1}  \sqrt{\log n}\, \log m \log\log m)
\times  \max_{u, v \in e} \|\bar u - \bar v \|^2.
\]
For 
$D = O(\beta^{-1}  \sqrt{\log n}\, \log m \log\log m)   /\alpha $
we get 
\[
\Prob{e \text{ is cut by } S} \leq \alpha D \max_{u, v \in e} \|\bar u - \bar v \|^2.
\]
Note that $\alpha \geq 1/2^l \geq \Omega(1/m)$. Thus
\[
D
\leq O(\beta^{-1}  \sqrt{\log n}\, m \log m \log\log m).
\]
\end{proof}

\newpage 
\appendix

\section{\texorpdfstring{$\ell_2$--$\ell_2^2$}{Euclidean} Hypergraph Orthogonal Separators}
\label{sec:orthogonal-l2-l22}
In this section, we present another variant of hypergraph orthogonal separators,
which we call $\ell_2$--$\ell_2^2$ hypergraph orthogonal separators.
The advantage of $\ell_2$--$\ell_2^2$ hypergraph orthogonal separators is
that their distortions do not depend on $n$ (the number of vertices).
Then in \prettyref{sec:algoHSSE-2}, we use $\ell_2$--$\ell_2^2 $
hypergraph orthogonal separators 
to prove \prettyref{thm:hsse-3} (which, in turn, 
implies \prettyref{thm:hsse-2}).
\begin{definition}[$\ell_2$--$\ell_2^2$ Hypergraph Orthogonal Separator]
Let $\set{\bar u:u\in V}$ be a set of vectors in the unit ball. We say that a random set $S\subset V$ is  \textit{a $\ell_2$--$\ell_2^2$ hypergraph $m$-orthogonal separator}
with $\ell_2$--distortion $D_{\ell_2}:{\mathbb N} \to {\mathbb R}$,
$\ell_2^2$--distortion $D_{\ell_2^2}$, probability scale $\alpha > 0$, and separation threshold $\beta \in (0,1)$ if it satisfies the following properties.
\begin{enumerate}
\item For every $u\in V$, 
\[\Pr(u\in S) = \alpha \|\bar u\|^2.\]
\item For every $u$ and $v$ such that $\| \bar u - \bar v \|^2 \geq \beta \min(\|\bar u\|^2, \|\bar v\|^2)$
% (i.e., the angle between $\bar u$ and $\bar v$ is at least ...).
\[ 
\Prob{u\in S \text{ and } v\in S} \leq \alpha \frac{\min (\|\bar u\|^2, \|\bar v\|^2)}{m}.
\]
\item For every $e\subset V$,
\[
\Prob{e \text{ is cut by } S} \leq \alpha D_{\ell_2^2} \cdot \max_{u, v \in e} \|\bar u - \bar v \|^2 + \alpha D_{\ell_2}(|e|) \cdot \min_{w\in e}{\|bar w\|} \cdot \max_{u, v \in e} \|\bar u - \bar v \|.
\]
\end{enumerate}
(This definition differs from \prettyref{def:h-othogonal-sep} only in item 3.)
\end{definition}
\begin{theorem}
\label{thm:hyper-orth-sep-l2}
There is a polynomial-time randomized algorithm that given a set of vertices $V$, a set of vectors $\{\bar u\}$ satisfying $\ell_2^2$--triangle inequalities, and parameters
$m$ and $\beta$ generates an $\ell_2$--$\ell_2^2$ hypergraph $m$-orthogonal separator with probability scale $\alpha  \geq 1 / n$ 
and distortions:
\begin{align*}
D_{\ell_2^2} &= O(m), \\
D_{\ell_2}(r) &= O(\beta^{-1/2} \sqrt{\log r}\, m\log m \log\log m).
\end{align*}
Note that distortions $D_{\ell_2^2}$ and $D_{\ell_2}$ do not depend on $n$.
\end{theorem}

The algorithm and its analysis are very similar to those in the proof of \prettyref{thm:hyper-orth-sep}.  The only difference is that we use another procedure to generate random assignments $\omega:V\to \set{0,1}$.
The following lemma is an analog of \prettyref{lem:omega}.
\begin{lemma}\label{lem:omega-l2} 
There is a randomized polynomial time algorithm that given a finite set $V$, vectors $\varphi(\bar u)$ for $u\in V$, satisfying $\ell_2^2$ triangle inequalities, 
and a parameter $\beta\in(0,1)$, returns a random assignment 
$\omega:V \to \set{0,1}$  that satisfies the following properties.
\begin{itemize}
\item For every set $e\subset V$ of size at least 2, 
\[
\Prob{\omega(u)\neq \omega(v) \text{ for some } u,v\in e} \leq O(\beta^{-1/2}  \sqrt{\log |e|}) \times \max_{u,v\in e} \|\varphi(\bar u) - \varphi(\bar v)\|.
\]

\item  For every $u$ and $v$ such that $\| \varphi(\bar u) - \varphi(\bar v) \|^2 \geq \beta$,
\[
\Prob{\omega(u)\neq \omega(v)} \geq 0.3.
\]
\end{itemize}
\end{lemma}
\begin{proof}
We sample a random Gaussian vector $g\sim {\cal N}(0,I_n)$ (each component $g_i$
of $g$ is distributed as ${\cal N}(0,1)$, all random variables $g_i$ are mutually independent).
Let $N$ be a Poisson process on $\mathbb R$ with rate $1/\sqrt{\beta}$.
Let 
$w(u)=1$ if $N(\iprod{g}{u})$ is even, and $w(u)=0$ if $N(\iprod{g}{\varphi(\bar u)})$ is odd.
Note that $\omega(u) = \omega(v)$ if and only if $N(\iprod{g}{\varphi(\bar u)}) - N(\iprod{g}{\varphi(\bar v)})$ is even. 

Consider a set $e\subset V$ of size at least 2. 
Denote $\diam(e) = \max_{u,v\in e} \|\varphi(\bar u)-\varphi(\bar v)\|$.
Let $\tau_m = \min_{w\in e} 
\iprod{g}{\varphi(\bar w)}$ and $\tau_M = \max_{w\in e} 
\iprod{g}{\varphi(\bar w)}$. Note that
\begin{align*}
N(\tau_m) &= \min_{w\in e} N(\iprod{g}{\varphi(\bar w)}),\\
N(\tau_M) &= \max_{w\in e} N(\iprod{g}{\varphi(\bar w)}).
\end{align*}
If all numbers $N(\iprod{g}{\varphi(\bar u)})$ are equal then $\omega(u) = \omega(v)$ for all $u,v \in e$. Thus if  $\omega(u) \neq \omega(v)$ for some $u,v\in e$ then $N(\iprod{g}{\varphi(\bar u)}) \neq N(\iprod{g}{\varphi(\bar v)})$ for some $u,v\in e$. In particular, then $N(\tau_M) - N(\tau_m) > 0$. 
Given $g$, $N(\tau_M) - N(\tau_m)$ is a Poisson random variable with rate
$(\tau_M - \tau_m)/\sqrt{\beta}$. We have,
\begin{align*}
\Prob{\omega(u)\neq \omega(v) \text{ for some } u,v\in e \given g} &
\leq \Prob{N(\tau_M) - N(\tau_m) > 0 \given g} \\&
{}= 1 -e^{-(\tau_M - \tau_m)/\sqrt{\beta}} \leq \beta^{-1/2}(\tau_M - \tau_m).
\end{align*}
Let $\xi_{uv} = \iprod{g}{\varphi(\bar u)} - \iprod{g}{\varphi(\bar v)}$
for $u,v\in e$ ($u\neq v$).
Note that $\xi_{uv}$ are Gaussian random variables with mean 0, and
\[
\Var[\xi_{uv}] =  \Var[\iprod{g}{\varphi(\bar u)} - \iprod{g}{\varphi(\bar v)}]
= \|\varphi(\bar u) - \varphi(\bar v)\|^2 \leq \diam(e)^2
\]
Note that the expectation of the maximum of (not necessarily independent) $N$ Gaussian random variables with standard deviation bounded by $\sigma$ is
$O(\sqrt{\log N} \sigma)$.
% That's obvious. 
% (e.g., that follows from Talagrand's Theorem or Slepian's Lemma).
We have,
\[
\E{\tau_M - \tau_m} = \E{\max_{u,v\in e} (\xi_{uv})} = O(\sqrt{\log |e|} \diam(e))
\]
since the total number of random variables $\xi_{uv}$ is $|e|(|e|-1)$.
Therefore,
\[
\Prob{\omega(u)\neq \omega(v) \text{ for some } u,v\in e}
\leq  \beta^{-1/2}\,\E{\tau_M - \tau_m} = O(\beta^{-1/2} \sqrt{\log |e|}\, \max_{u,v\in e} \|\varphi(\bar u)-\varphi(\bar v)\|).
\]
We proved that $\omega$ satisfies the first property.
Now we verify that $\omega$ satisfies the second condition. Consider two vertices $u$ and $v$ with  $\|\varphi(\bar u) - \varphi(\bar v) \|^2 \geq \beta$.
Given $g$, the random variable $Z = N(\iprod{g}{\varphi(\bar u)}) - N(\iprod{g}{\varphi(\bar v)})$ has Poisson distribution with rate $\lambda = |\iprod{g}{\varphi(\bar u)}) - \iprod{g}{\varphi(\bar v)}|/\sqrt{\beta}$. We have,
\[
\Prob{Z \text{ is even} \given g } = \sum_{k=0}^{\infty}
\Prob{Z = 2k \given g} = \sum_{k=0}^{\infty} \frac{e^{-\lambda}\lambda^{2k}}{(2k)!} = \frac{1 + e^{-2\lambda}}2.
\]
Note that $\lambda$ is the absolute value of a Gaussian random variable with mean 0 and standard deviation $\sigma = \|\varphi(\bar u) - \varphi(\bar v)\|/\sqrt{\beta} \geq 1$. Thus
\[
\Prob{Z \text{ is even}} = \E{\frac{1 + e^{-2\sigma |\gamma|}}2},
\]
where $\gamma$ is a standard Gaussian random variable, $\gamma \sim {\cal N}(0,1)$.
We have,
\[\Prob{\omega(u)\neq \omega(v)} = \E{\frac{1 - e^{-2\sigma |\gamma|}}2}
\geq \E{\frac{1 - e^{-2|\gamma|}}2} \geq 0.3.
\]
\end{proof}

Now we use the algorithm from \prettyref{thm:hyper-orth-sep}
to obtain $\ell_2$--$\ell_2^2$ hypergraph orthogonal separators. The only difference is that we use the procedure from \prettyref{lem:omega-l2} rather than from
\prettyref{lem:omega} to generate assignments $\omega$. We obtain a
$\ell_2$--$\ell_2^2$ hypergraph orthogonal separator.
\begin{theorem}
Random set $S$ is a hypergraph $m$-orthogonal separator with distortion
\begin{align*}
D_{\ell_2^2} &= O(m), \\
D_{\ell_2}(r) &= O(\beta^{-1/2} \sqrt{\log r}\, m\log m \log\log m),
\end{align*}
probability scale $\alpha \geq 1/n$ and separation threshold $\beta \in (0,1)$.
\end{theorem}
\begin{proof}
The proof of the theorem is almost identical to that of \prettyref{thm:h-ort-sep-distortion-bound}. We first check conditions 1 and 2
of $\ell_2$--$\ell_2^2$ hypergraph orthogonal separators in the same way as we checked conditions 1 and 2 of hypergraph orthogonal separators
in \prettyref{thm:h-ort-sep-distortion-bound}.
When we verify that property 3 holds, we use bounds from \prettyref{lem:omega-l2}. The only difference is how we upper bound the
probability of the event ${\cal E}_2$.

If ${\cal E}_2$ happens then (1) $r \leq \rho_m $ (since $A = e$) and (2) $W(u) \neq W(v)$ for some $u,v\in e$.
The probability that $r \leq \rho_m$ is $\rho_m$. We upper bound the probability that  $W(u) \neq W(v)$ for some  $u,v\in e$.
For each $i\in\set{1,\dots,l}$,
\begin{align*}
\Prob{\tilde\omega_i(u)\neq \tilde\omega_i(v) \text{ for some } u,v\in e} &\leq
O(\beta^{-1/2} \sqrt{\log |e|}\,  \log\log m) \max_{u, v \in e}  \|\varphi(\bar u) - \varphi(\bar v)\| \\ 
&\leq  
O(\beta^{-1/2} \sqrt{\log |e|}\,  \log\log m) \max_{u, v \in e} \frac{\|\bar u - \bar v \|}{\min(\|\bar u\|, \|\bar v\|)}\\
&\leq  O(\beta^{-1/2} \sqrt{\log |e|}\,  \log\log m)
\times \rho_m^{-1/2} \times  \max_{u, v \in e} \|\bar u - \bar v \|.
\end{align*}

By the union bound over $i\in\set{1,\dots,l}$, the probability that $W(u) \neq W(v)$ for some $u,v\in e$ is at most 
$O(l \times \beta^{-1/2} \sqrt{\log |e|}\,  \log\log m)
\times \rho_m^{-1/2} \times  \max_{u, v \in e} \|\bar u - \bar v \|$.
Therefore,
\begin{align*}
\Prob{{\cal E}_2} &\leq  \rho_m \times O(l \times \beta^{-1/2} \sqrt{\log |e|}\,  \log\log m) \times\rho_m^{-1/2} \times  \max_{u, v \in e} \|\bar u - \bar v \|\\
&\leq 
O(l \times \beta^{-1/2} \sqrt{\log |e|}\,  \log\log m) \times\rho_m^{1/2} \times  \max_{u, v \in e} \|\bar u - \bar v \|.
\end{align*}
We get that the probability that $e$ is cut by $S$ is at most
\begin{align*}
\Prob{{\cal E}_1} + \Prob{{\cal E}_2} &\leq 
\max_{u, v \in e} \|\bar u - \bar v \|^2 + O(l \times \beta^{-1/2} \sqrt{\log |e|}\,  \log\log m) \times\rho_m^{1/2} \times  \max_{u, v \in e} \|\bar u - \bar v \|\\
&\leq 
\max_{u, v \in e} \|\bar u - \bar v \|^2 + O(l \times \beta^{-1/2} \sqrt{\log |e|}\,  \log\log m) \times\min_{w\in e} \|\bar w\| \times  \max_{u, v \in e} \|\bar u - \bar v \|.
\end{align*}
For 
$D_{\ell_2^2} = 1/\alpha$ and $D_{\ell_2}(r) = O(\beta^{-1/2} \sqrt{\log r}\,  \log m \log\log m)/\alpha$,
we get 
\[
\Prob{e \text{ is cut by } S} \leq \alpha D_{\ell_2^2} \cdot \max_{u, v \in e} \|\bar u - \bar v \|^2 + \alpha D_{\ell_2}(|e|) \cdot \min_{w\in e}\|\bar w\|\cdot \max_{u, v \in e} \|\bar u - \bar v \|.
\]
Note that $\alpha \geq 1/2^l \geq \Omega(1/m)$. Thus
\begin{align*}
D_{\ell_2^2} &= O(m), \\
D_{\ell_2}(r) &= O(\beta^{-1/2} \sqrt{\log r}\, m\log m \log\log m).
\end{align*}
\end{proof}

\section{Algorithm for Hypergraph Small Set Expansion via 
\texorpdfstring{$\ell_2$--$\ell_2^2$}{Euclidean} Hypergraph Orthogonal Separators}\label{sec:algoHSSE-2}
In this section, we present another algorithm for Hypergraph Small Set Expansion. The algorithm finds a set with expansion proportional
to ${\sqrt{\phi^*_{G,\delta}}}$. The proportionality constant depends on degrees of vertices and hyperedge size but not on 
the graph size. Here, we present our result for arbitrary hypergraphs. The result for uniform hypergraphs (\prettyref{thm:hsse-2}) stated in the introduction follows
from our general result. In order to state our result for arbitrary graphs, we need the following definition.

\begin{definition}\label{def:agdeg}
Consider a hypergraph $H = (V, E)$. Suppose that for every edge $e$ we are given a non-empty subset $e^{\circ} \subseteq e$.% (possibly, $e^{\circ} = e$). 
%Let $\deg^{\circ} u = \|\set{e: u\in e^{\circ}}\|$. 
Let 
\begin{align*}
\eta(u) &= \sum_{e: u\in e^{\circ}} \frac{\log_2 |e|}{|e^{\circ}|},\\
\eta_{max} &= \max_{u\in V} \eta(u).
\end{align*}
Finally, let $\agdeg$ be the minimum of $\eta_{max}$ over all possible choices of subsets $e^{\circ}$.
\end{definition}

\begin{claim}\label{claim:agdeg}
\begin{enumerate} 
\item $\agdeg \leq \max_{u\in V} \sum_{e:u\in e} (\log_2 |e|)/|e|$.
\item If $H$ is a $r$-uniform graph with maximum degree $d_{\text{max}}$ then $\agdeg \leq (d_{\text{max}} \log_2 r)/r$.
\item Suppose that we can choose one vertex in every edge so that no vertex is chosen more than once. Then $\agdeg \leq \log_2 r_{\text{max}}$,
where $r_\text{max}$ is the size of the largest hyperedge in $H$.
\end{enumerate} 
\end{claim}
\begin{proof}\hspace*{0pt}
\begin{enumerate}
\item Let $e^{\circ} = e$ for every $e\in E$. We have, $\agdeg \leq \max_{u\in V} \sum_{e:u\in e} (\log_2 |e|)/|e|$.

\item By item 1, $\agdeg \leq \max_{u\in V} \sum_{e:u\in e} (\log_2 |e|)/|e| = \max_{u\in V} \sum_{e:u\in e} (\log_2 r)/r = (d_{\text{max}} \log_2 r )/r$.

\item For every edge $e\in E$, let $e^{\circ}$ be the set that contains the vertex chosen for $e$. Then $|e^{\circ}| = 1$ and $|\{e:u\in e^{\circ}\}| \leq 1$ for every $u$. We have,
\[
\agdeg \leq \max_{u\in V} \sum_{e: u\in e^{\circ}} \frac{\log_2 |e|}{|e^{\circ}|} \leq  \max_{u\in V} \sum_{e: u\in e^{\circ}} \frac{\log_2 r_{\text{max}}}{1}
= \log_2 r_{\text{max}}.
\]
\end{enumerate}
\end{proof}

\begin{theorem}\label{thm:hsse-3}
There is a randomized polynomial-time algorithm for the Hypergraph Small Set Expansion problem that
given a hypergraph $H = (V,E)$, and parameters $\varepsilon \in (0,1)$ and $\delta \in (0,1/2]$,
finds a set $S \subset V$ of size at most $(1+\varepsilon)\delta n$ such that 
\begin{align*}
\phi(S) &\leq O_{\varepsilon}\left(\delta^{-1}\log \delta^{-1}\log\log \delta^{-1}\,
\sqrt{\agdeg \cdot \phi^*_{H,\delta}} +
\delta^{-1} \phi^*_{H,\delta}\right) \\
&= \tilde O_{\varepsilon} \left(\delta^{-1} \left(\sqrt{\agdeg \phi^*_{H,\delta}} + \phi^*_{H,\delta}\right)\right),
\end{align*}
In particular, if $H$ is an $r$-uniform hypergraph with maximum degree $d_{\text{max}}\,$,
then we have,
\[
\phi(S) \leq \tilde O_{\varepsilon} \left(\delta^{-1} \left(\sqrt{ 
d_{\text{max}} \frac{\log_2 r}{r}
\phi^*_{H,\delta}} + \phi^*_{H,\delta}\right)\right).
\]
\end{theorem}
\begin{proof}
The proof is similar to that of \prettyref{thm:HSSE-proof}.
We solve the SDP relaxation for H-SSE and obtain an SDP solution $\set{\bar u}$. Denote the SDP value by $\sdpcost$.
Consider an $\ell_2$--$\ell_2^2$ hypergraph orthogonal separator $S$
with $m = 4/(\varepsilon \delta)$ and $\beta = \varepsilon/4$.
Define a set $S'$: 
\[
S' = 
\begin{cases}
S,\ \text{if } |S| \leq (1+\varepsilon) \delta n,\\
\varnothing,\ \text{otherwise.}
\end{cases}
\]
Clearly, $|S'| \leq (1+\varepsilon) \delta n$.
As in the proof of \prettyref{thm:HSSE-proof},
$\Pr(u \in S') \in \bigl[\frac{\alpha}{2} \, \|\bar u\|^2,\alpha 
\|\bar u\|^2\bigr]  $. 
Note that 
\[\Prob{S' \text{ cuts edge } e} \leq \Prob{S \text{ cuts edge } e}
\leq \alpha D_{\ell_2} \max_{u,v\in e} \|\bar u - \bar v\|^2 + \alpha D_{\ell_2}(r) \min_{w\in e} \|\bar w\|\max_{u,v\in e} \|\bar u - \bar v\|.\]
Let ${\cal C} = \alpha^{-1}\E{|\Ecut{S'}|} $. 
Let $Z = |S'| - \frac{|\Ecut{S'}|}{4{\cal C}}$.
We have,
\[ \E{Z} = \E{|S'|} - \E{\frac{|\Ecut{S'}|}{4{\cal C}}}
\geq \sum_{u \in V} \frac{\alpha}{2} \cdot \|\bar u\|^2
= \frac{\alpha}{2} - \frac{\alpha}{4}= \frac{\alpha}{4}.\]
Now we upper bound $\cal C$. Consider the optimal choice of $e^{\circ}$ for $H$ in the definition of $\agdeg$. 
\begin{align*}
{\cal C}=\alpha^{-1} \E{|\Ecut{S'}|}  
&\leq
\alpha^{-1} \sum_{e\in E} \Prob{e \text{ is cut by } S}\\
&\leq
D_{\ell_2^2} \sum_{e\in E} \max \|\bar u
-\bar v\|^2 + \sum_{e\in E} D_{\ell_2}(|e|) \min_{w\in e} \|\bar w\| \max_{u,v\in e} \|\bar u-\bar v\|\\
& \leq D_{\ell_2^2} \cdot \sdpcost 
 + \sum_{e\in E} D_{\ell_2}(|e|) \, \sum_{w\in e^{\circ}} \left(\frac{\|\bar w\| }{|e^{\circ}|}\right) \times \max_{u,v\in e} \|\bar u-\bar v\| \\
& \leq D_{\ell_2^2} \cdot \sdpcost 
 + \sum_{e\in E} \sum_{w\in e^{\circ}} \frac{D_{\ell_2}(|e|) \|\bar w\| }{\sqrt{|e^{\circ}|}} \times \frac{\max_{u,v\in e} \|\bar u-\bar v\|}{\sqrt{|e^{\circ}|}} \\
\text{\footnotesize (by Cauchy---Schwarz) }
& \leq D_{\ell_2^2} \cdot \sdpcost 
 + \sqrt{\sum_{e\in E} \sum_{w\in e^{\circ}} \frac{D_{\ell_2}(|e|)^2 \|\bar w\|^2 }{|e^{\circ}|}}
   \sqrt{\sum_{e\in E} \sum_{w\in e^{\circ}}  \frac{\max_{u,v\in e} \|\bar u-\bar v\|^2}{|e^{\circ}|} }\\
&\leq D_{\ell_2^2} \cdot \sdpcost 
 + \sqrt{\sum_{w\in V} \sum_{e: w\in e^{\circ}} \frac{D_{\ell_2}(|e|)^2  }{|e^{\circ}|} \|\bar w\|^2}\,
\sqrt{\sdpcost} .
\end{align*}
For every vertex $w$,
\[
\sum_{e: w\in e^{\circ}} \frac{D_{\ell_2}(|e|)^2  }{|e^{\circ}|} \leq
O_{\beta}(m\log m\log\log m)^2 {\sum_{e: w\in e^{\circ}} \frac{\log_2 |e|}{|e^{\circ}|}} \leq O_{\beta}(m\log m\log\log m)^2 \times {\agdeg}.
\]
and $\sum_{w\in V} \|\bar w\|^2 = 1$.
Therefore,
\[
{\cal C} \leq O_{\beta}\left(m \sdpcost + m\log m\log\log m\,
\sqrt{\agdeg \cdot \sdpcost}\right). 
\]
By the argument from \prettyref{thm:HSSE-proof},
we get that if we sample $S'$ sufficiently many times (i.e., ($4n^2/\alpha$) times),
we will find a set $S'$ such that
\[
\frac{|\Ecut{S'}|}{|S'|} \leq 4{\cal C}\leq 
O_{\beta}\left(\delta^{-1}\log \delta^{-1}\log\log \delta^{-1}\,
\sqrt{\agdeg \cdot \sdpcost} +
\delta^{-1} \sdpcost\right)
\]
with probability  exponentially close to 1.
\end{proof}
\section{Reduction from Vertex Expansion to Hypergraph Expansion}
\label{sec:vertex-to-hyper}

In the reduction from vertex expansion to hypergraph expansion, we will use the notion of 
{\em Symmetric Vertex Expansion}.
For a graph $G = (V,E)$, and for a set $S \subset V$, we define its internal neighborhood $\Nin(S)$,  
and its outer neighborhood $\Nout(S)$ as follows.
\[ \Nin(S) =  \{ u \in S : \exists\, v \in \bar{S} \textrm{ such that } \{u,v \} \in E \} \] 
\[ \Nout(S) = \{ u \in \bar{S} : \exists\, v \in S \textrm{ such that } \{u,v \} \in E \}  .  \]
The symmetric vertex expansion of a set, denoted by
$\Phi^V(S)$, is defined as
\begin{align*}
 \Phi^V(S) &=  \frac{ | \Nin(S)  \cup \Nout(S) | }{ \min (|S|,|\bar S|)  } ,\\
 \Phi^V_{G,\delta} &= \min_{\substack{S\subset V \\
 0<|S|\leq \delta n}} \Phi^V(S). 
\end{align*}
We will use the following reduction from vertex expansion to symmetric vertex expansion. 

\begin{theorem}[\cite{lrv13}]
\label{thm:sym-vert}
Given a graph $G$, there exists a graph $G'$  such that 
\[  c_1 \phi^V_{G,\delta} \leq    \Phi^V_{G',\delta}  \leq c_2  \phi^V_{G,\delta} . \]
where $c_1,c_2 > 0$ are absolute constants, and 
the maximum degree of graph $G'$ is equal to the maximum degree of graph $G$.
Moreover, there exists a polynomial time algorithm to compute such graph $G'$.
\end{theorem}

\begin{theorem}[Restatement of \prettyref{thm:hyper-vert-exp}]
There exist absolute constants $c_1', c_2' \in \mathbb{R}^+$ such that for 
 every  graph $G=(V,E)$, %of maximum degree $d$, 
there exists a polynomial time computable hypergraph 
$H = (V',E')$ %having the hyperedges of cardinality at most $d+1$ 
such that
\[ c_1' \phi^*_{H,\delta} \leq \phi^V_{G,\delta} \leq c_2' \phi^*_{H,\delta}, \]
and $\agdeg \leq \log_2 d_{\text{max}}$.
\end{theorem}

\begin{proof}
Starting with graph $G$,
we use \prettyref{thm:sym-vert} to obtain a graph $G' = (V',E')$ such that 
\begin{equation}
 c_1 \phi^V_{G,\delta} \leq    \Phi^V_{G',\delta}  \leq c_2  \phi^V_{G,\delta} . 
\label{eq:sym-expansion}
\end{equation}

Next we construct hypergraph $H = (V',E'')$ as follows. For every vertex $v \in V'$, we add the
hyperedge $\{v \}  \cup \Nout(\{v\})$ to $E''$ (note that $\Nout(\{v\})$ is the set of neighbors of $v$ in $G$).
%By construction, all hyperedges in $E'$ have cardinality at most $d+1$.
Fix an arbitrary set $S \subset V$.

We first show that $\Phi^V(S) \leq  \phi_H(S) $. 
Consider the vertices $\Nin(S)$. Each vertex in $v \in \Nin(S)$ 
has a neighbor, say $u$, in $\bar{S}$. Therefore  
the hyperedge $\{v\} \cup \Nout(\{v\})$ is cut by $S$ in $H$. Similarly, 
for each vertex $v \in \Nout(S)$, the hyperedge $\{v\} \cup \Nout(\{v\})$ is cut by 
$S$ in $H$.
All these hyperedges are disjoint by construction. Therefore,
\[ \Phi^V(S) = \frac{ \Abs{\Nin(S)} + \Abs{\Nout(S)} }{ \Abs{S}  } \leq 
  \frac{ \Abs{E_{cut}(S)}   }{ \Abs{S}} \leq  \phi_H(S) .
\]

Now we verify that $\phi_H(S) \leq \Phi^V(S) $.
For any hyperedge $( \{v\} \cup \Nout(\{v\}) ) \in E_{cut}(S)$, 
the vertex $v$ has to belong to either $\Nin(S)$ or $\Nout(S)$. Therefore,
\[ \phi_H(S) \leq \frac{ \Abs{E_{cut}(S)} }{ \Abs{S}  } \leq   
	 \frac{ \Abs{\Nin(S)} + \Abs{\Nout(S)} }{ \Abs{S}  } = \Phi^V(S) . \]
Therefore, we get that 
\begin{align*} \phi_H(S) &= \Phi^V(S) \qquad \makebox[0pt][l]{for every $S \subset V$,}   \\
\intertext{and hence}
\phi_{H,\delta}^* &= \Phi_{G',\delta}^V
\end{align*}
We get from \prettyref{eq:sym-expansion},
%\[ c_1 \phi^*_{H,\delta} \leq \phi^V_{G,\delta} \leq c_2 \phi^*_{H,\delta}. \]
\[ c_1 \phi^V_{G,\delta} \leq \phi^*_{H,\delta}   \leq  c_2\phi^V_{G,\delta}. \]

Finally, we upper bound $\agdeg$. We use part 3 of \prettyref{claim:agdeg}. We choose vertex $v$ in the hyperedge $\set{v} \cup \Nout(\set{v})$. 
By \prettyref{claim:agdeg}, $\agdeg \leq \log_2 r_{\text{max}}$, where  $r_{\text{max}}$ is the size of the largest hyperedge.
Note that $|\set{v} \cup \Nout(\set{v})| = \deg v + 1$. Thus  $\agdeg \leq \log_2 r_{\text{max}} \leq \log_2 (d_{\text{max}} + 1)$
\end{proof}

\section{SDP Intgrality Gap} \label{sec:SDPgap}
In this section, we present an integrality  gap for the SDP relaxation for H-SSE. We also give  a lower bound on the distortion of a hypergraph $m$-orthogonal separator.

\begin{theorem}\label{thm:gap}
For $\delta = 1/r$, the integrality gap of the SDP for H-SSE is at least $1/(2\delta) = r/2$.
\end{theorem}
\begin{proof}
Consider a hypergraph $H= (V,E)$ on $n=r$ vertices with one hyperedge $e = V$ ($e$ contains all vertices). Note that the expansion of every set of size $\delta n = 1$ 
is $1$. Thus $\phi_{H,\delta}^* = 1$. 

Consider an SDP solution that assigns vertices mutually orthogonal vectors of length $1/\sqrt{r}$. It is easy to see this is a feasible SDP solution.
Its value is $\max_{u,v\in e} \|\bar u - \bar v\|^2 = 2/r$. Therefore, the SDP integrality gap is at least $r/2$.
\end{proof}

Now we give a lower bound on the distortion of hypergraph $m$-orthogonal separators.
\begin{lemma}
For every $m > 4$, there is an SDP solution such that every hypergraph $m$-orthogonal separator with 
separation threshold $\beta \geq 0$
has distortion at least $\lceil m\rceil/4$.
\end{lemma}
\begin{proof}
Consider the SDP solution from \prettyref{thm:gap} for $n=r=\lceil m\rceil$. Consider a hypergraph $m$-orthogonal separator $S$ for this solution.
Let $D$ be its distortion. Note that condition (2) from the definition of hypergraph orthogonal separators applies
to any pair of distinct vertices $(u,v)$ since $\iprod{\bar u}{\bar v} = 0$.

By the inclusion--exclusion principle, we have,
\begin{align*}
\Prob{|S| = 1} &\geq \sum_{u\in S} \Prob{u\in S} - \frac12 \sum_{u,v\in S, u\neq v} \Prob{u\in S,v\in S}\\
&\geq \sum_{u\in S} \alpha \|\bar u\|^2  - \frac12 \sum_{u,v\in S, u\neq v} \frac{\alpha \min(\|\bar u\|^2,\|\bar v\|^2)}{m}\\
&= \alpha - \frac{\alpha n(n-1)}{2mr} = \alpha \left(1 - \frac{(n-1)}{2m}\right) \geq \alpha /2.
\end{align*}
On the other hand, if $|S| = 1$ then $S$ cuts $e$. We have, 
\[
\Prob{|S| = 1} \leq \Prob{S \text{ cuts } e} \leq \alpha D \max_{u,v\in e} \|\bar u - \bar v\|^2 = 2\alpha D/r.
\]
We get that $\alpha /2 \leq 2\alpha D/r$ and thus $D \geq r/4 = \lceil m\rceil/4$.
\end{proof}

\section{Proof of \prettyref{lem:omega-amplified}}
\label{sec:proof-omega-amplified}
%\begin{proof}
Let $K = \max\left(\left\lceil \frac{\log_2 \log_2 m}{-\log_2(1-4p)}\right\rceil,1\right)$. We independently sample $K$ assignments $\omega_1,\dots,
\omega_K$. Let 
\[
\tilde\omega(u) = \omega_1(u) \oplus \dots \oplus \omega_K(u),
\]
where $\oplus$ denotes addition modulo 2.
Consider $u$ and $v$ such that $\| \varphi(\bar u) - \varphi(\bar v) \|^2 \geq \beta$. Let 
$$\tilde p = \Prob{\omega_i (u) \neq \omega_i(v)} \geq 2p \quad \text{for}\quad i\in\set{1,\dots, K}$$
(the expression does not depend on the value of $i$ since all $\omega_i$
are identically distributed). Note that $\tilde\omega(u) \neq \tilde\omega(v)$
if and only if $\omega_i (u) \neq \omega_i(v)$ for an odd number of values $i$. Therefore,
\begin{align*}
\Prob{\omega(u) \neq \omega(v)} &= 
\sum_{0\leq k\leq K/2} \binom{K}{2k+1} \tilde p^{2k+1} (1-\tilde p)^{K-2k-1}=
\frac{1 - (1-2\tilde p)^K}{2}\\
&\geq \frac{1 - (1-4 p)^K}{2}
\geq \frac{1}{2} - \frac{1}{\log_2 m}.
\end{align*}
Now let $e\subset V$ be a subset of size at least 2. We have,
\[
\Prob{\tilde \omega (u) \neq \tilde \omega (v)} 
\leq \Prob{\omega_i (u) \neq \omega_i (v) \text{ for some } i}
\leq O(K\beta^{-1}  \sqrt{\log n}\max_{u,v\in e} \|\varphi(\bar u) - \varphi(\bar v)\|^2).
\]
%\end{proof}

\end{document}